\documentclass[letterpaper, 10 pt, conference]{ieeeconf}  

\IEEEoverridecommandlockouts
\usepackage{epsfig} 
\usepackage[font={small}]{caption}
\usepackage{graphicx}
\usepackage{xcolor}
\usepackage{comment}
\usepackage{url}
\newcounter{num}

\usepackage{amsmath, amssymb}
\usepackage{subcaption}
\usepackage{algorithm}
\usepackage[noend]{algpseudocode}
\usepackage{color}
\usepackage{mathtools}
\usepackage{multirow}

\newtheorem{theo}{Theorem}
\newtheorem{defi}{Definition}

\newtheorem{prop}{Proposition}
\newtheorem{example}{Example}

\newtheorem{remark}{Remark}

\algrenewcommand\algorithmicrequire{\textbf{Input:}}
\algrenewcommand\algorithmicensure{\textbf{Output:}}

\algnewcommand{\Break}{\textbf{break}}

\algnewcommand\algorithmicforeach{\textbf{for each}}
\algdef{S}[FOR]{ForEach}[1]{\algorithmicforeach\ #1\ \algorithmicdo}

\makeatletter
\newcommand*{\itemequation}[3][]{%
  \item
  \begingroup
    \refstepcounter{equation}%
    \ifx\\#1\\%
    \else  
      \label{#1}%
    \fi
    \sbox0{#2}%
    \sbox2{$\displaystyle#3\m@th$}%
    \sbox4{\@eqnnum}%
    \dimen@=.5\dimexpr\linewidth-\wd2\relax
    \ifcase
        \ifdim\wd0>\dimen@
          \z@
        \else
          \ifdim\wd4>\dimen@
            \z@
          \else 
            \@ne
          \fi 
        \fi
      \@latex@warning{Equation is too large}%
    \fi
    \noindent   
    \rlap{\copy0}%
    \rlap{\hbox to \linewidth{\hfill\copy2\hfill}}%
    \hbox to \linewidth{\hfill\copy4}%
    \hspace{0pt}
  \endgroup
  \ignorespaces 
}
\makeatother

\makeatletter
\def\@citex[#1]#2{\leavevmode
\let\@citea\@empty
\@cite{\@for\@citeb:=#2\do
{\@citea\def\@citea{,\penalty\@m\ }%
\edef\@citeb{\expandafter\@firstofone\@citeb\@empty}%
\if@filesw\immediate\write\@auxout{\string\citation{\@citeb}}\fi
\@ifundefined{b@\@citeb}{\hbox{\reset@font\bfseries ?}%
\G@refundefinedtrue
\@latex@warning
{Citation `\@citeb' on page \thepage \space undefined}}%
{\@cite@ofmt{\csname b@\@citeb\endcsname}}}}{#1}}
\makeatother

\title{\LARGE \bf
Signal Temporal Logic Meets Convex-Concave Programming: \\ A Structure-Exploiting SQP Algorithm for STL Specifications}

\IEEEaftertitletext{\vspace{-1\baselineskip}}
\author{Yoshinari~Takayama$^{1}$, Kazumune Hashimoto$^{2}$, and Toshiyuki Ohtsuka$^{3}$
\thanks{}
\thanks{$^{1}$Y. Takayama is with Laboratoire des Signaux et Systèmes, Université Paris-Saclay, CNRS, CentraleSupélec, Gif-sur-Yvette, France.}
\thanks{\tt\small yoshinari.takayama@l2s.centralesupelec.fr} 
\thanks{$^{2}$K. Hashimoto is with the Graduate School of Engineering, Osaka University,
Suita, Japan. \tt\small{hashimoto@eei.eng.osaka-u.ac.jp}}
\thanks{$^{3}$T. Ohtsuka are with the Department of Systems Science, Graduate School of Informatics, Kyoto University, Kyoto, Japan. \tt\small{ohtsuka@i.kyoto-u.ac.jp}}
\thanks{This work was partially supported by JSPS KAKENHI under Grant JP22H01510 and 21K14184 and by JST CREST under Grant JPMJCR201. }
}

\begin{document}
\IEEEaftertitletext{\vspace{-1\baselineskip}}

\maketitle

\begin{abstract}
This study considers the control problem with signal temporal logic (STL) specifications. 
Prior works have adopted smoothing techniques to address this problem within a feasible time frame and solve the problem by applying sequential quadratic programming (SQP) methods naively. However, one of the drawbacks of this approach is that solutions can easily become trapped in local minima that do not satisfy the specification.
In this study, we propose a new optimization method, termed \textit{CCP-based SQP}, based on the convex-concave procedure (CCP). Our framework includes a new \textit{robustness decomposition} method that decomposes the robustness function into a set of constraints, resulting in a form of difference of convex (DC) program that can be solved efficiently. We solve this DC program sequentially as a quadratic program by only approximating the disjunctive parts of the specifications. Our experimental results demonstrate that our method has a superior performance compared to the state-of-the-art SQP methods in terms of both robustness and computational time. 

\end{abstract}

\section{Introduction}
As technology becomes more and more powerful, autonomous systems––such as drones and self-driving cars––are required to execute increasingly complex control tasks.
Signal temporal logic (STL) has been progressively used as a specification language for these complex robotic tasks in various situations, due to their expressiveness and closeness to natural language \cite{Fainekos2009-yn}. However, the resulting optimization problem is a mixed integer program (MIP) \cite{karaman2008,Raman2014-sa}, which can easily be computationally intractable and scales poorly with the number of integer variables. To avoid this complexity, recent work formulates the problem as nonlinear programs (NLP) and solve it by sequential quadratic programming (SQP) methods \cite{Pant_smooth,Gilpin2021-wv}. Although these methods can find a locally optimal solution in a short time compared to the MIP-based methods in general, the major difficulty of this formulation is that the solution can get stuck in a local minimum far from the global optimum as this method approximates the highly non-convex programs to quadratic programs. Moreover, the popular formulation makes this method more inefficient because all nonconvex functions are compressed into the objective function––which results in a coarse approximation by the naive SQP method.

In this study, we address the aforementioned challenges by leveraging the structure of STL control problems. We introduce a novel optimization framework, termed the \textit{convex-concave procedure based sequential quadratic programming (CCP-based SQP)}, which is a variant of the SQP method that solves quadratic programs sequentially. However, unlike traditional SQP methods, this framework approximates the original program at each iteration differently. It is based on the iterative optimization approach of CCP~\cite{Lipp2016-fa,Shen2016-sn}. Stated differently, only the concave portion is linearized at each iteration, allowing the method to retain the complete information of the convex parts. This is in contrast to naive SQP, which can only retain second-order derivative information at each iteration. As a result, our approach was able to synthesize a trajectory that is not only more robust (in terms of robustness score) but is also produced in a shorter time than the naive SQP methods.

\begin{figure*}[t]

 \begin{minipage}[b]{0.25\linewidth}
    \centering
\includegraphics[keepaspectratio, scale=0.29]{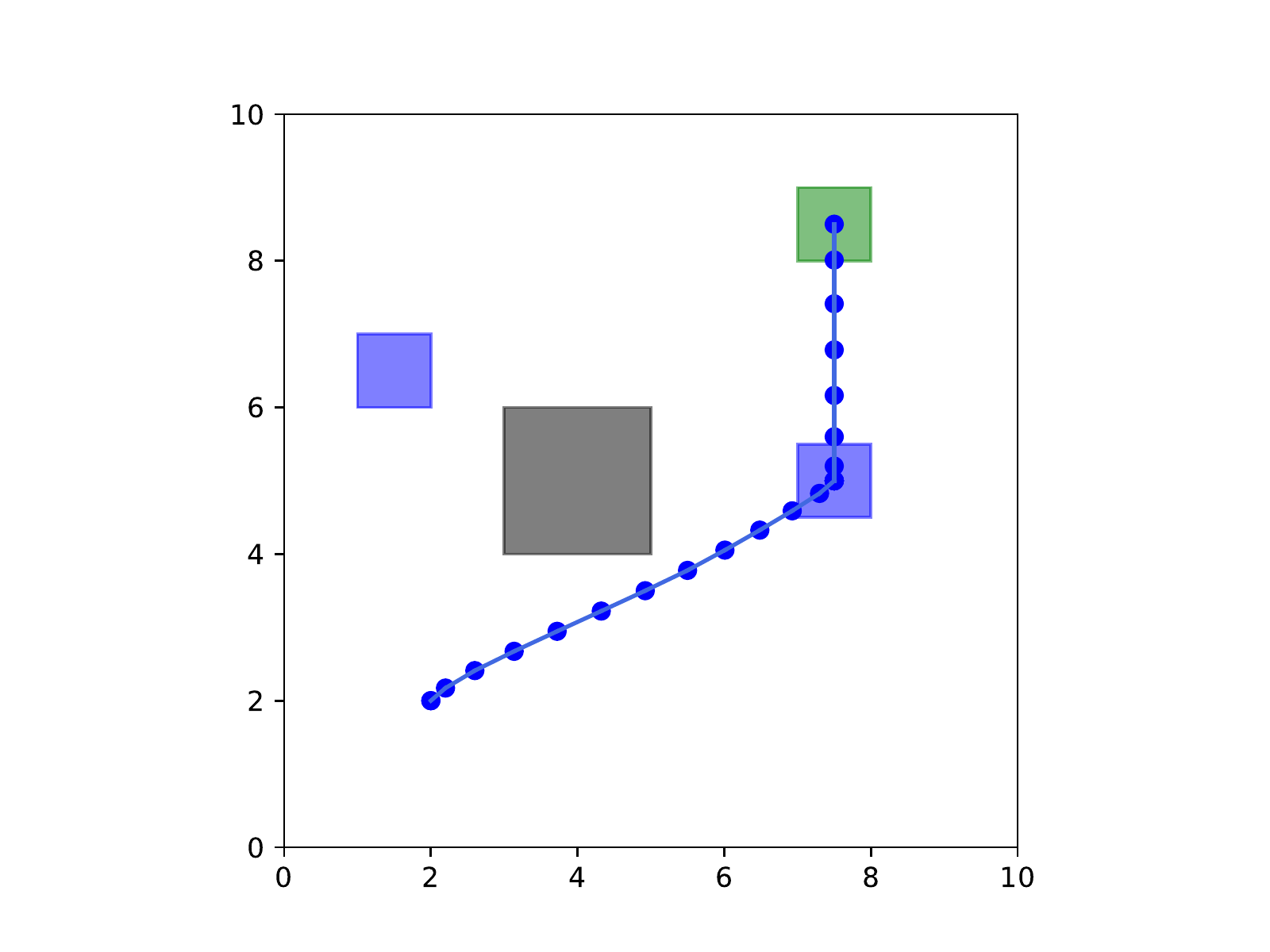}
    \caption{Two-target scenario}\label{fig:either-MIP}
 \end{minipage}
  \begin{minipage}[b]{0.75\linewidth}
    \includegraphics[keepaspectratio, scale=0.7]{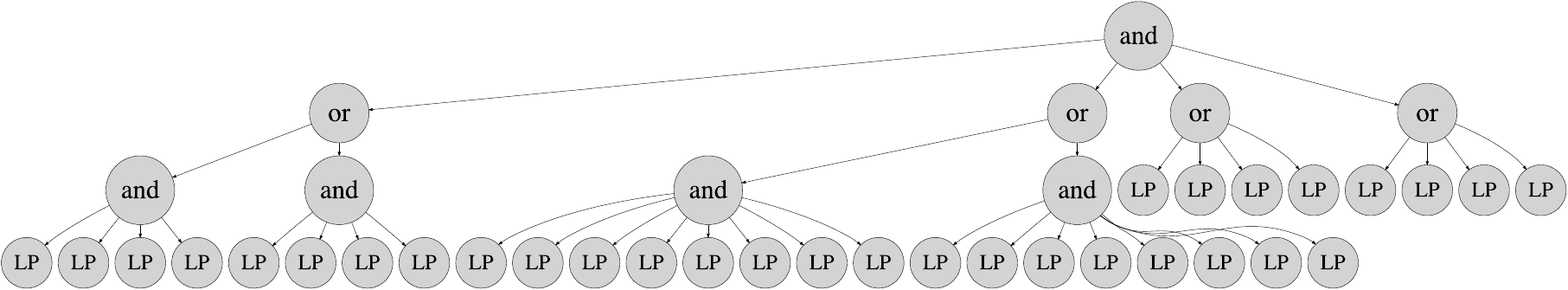}
    \caption{A tree description of the two-target formula with $T=1$ and $T_d=1$ }\label{fig:tree_two_target}
 \end{minipage}
  \end{figure*}

The remainder of this paper is organized as follows. We first provide the notation, preliminaries, and problem formulation in Section \ref{section:problemstatement}. Then, Section~\ref{section:decomposition} provides our robustness decomposition method that transforms the optimization problem to an efficient DC form. Section \ref{sec:subproblem} provides some properties of the resulting CCP-based SQP framework. Section \ref{section:example} demonstrates the effectiveness of our method over a state-of-the-art naive SQP method.

\section{Preliminaries}\label{section:problemstatement}

$\mathbb{R}$, and $\mathbb{Z}$ are defined as the fields of real and integer numbers, respectively. Given $a, b \in \mathbb{Z}$ with $a < b$, $[a, b]$ denotes a set of integers from $a$ to $b$. True and false are denoted by $\top$ and $\perp$. $I_n$ denotes the identity matrix of size $n$. $0_{n\times m}$ denotes the zero matrix of size $n\times m$.

\subsection{System description}
Throughout this paper, we consider a discrete-time linear system:
\begin{equation}\label{eq:system}
x_{t+1} =A x_t+B u_t, 
\end{equation}
where $A\in\mathbb{R}^{n\times n},B\in\mathbb{R}^{n\times m},$ and $x_t \in \mathcal{X} \subseteq \mathbb{R}^n, u_t \in \mathcal{U} \subseteq$ $\mathbb{R}^m$ denote the state and input. We assume that $\mathcal{X}$ and $\mathcal{U}$ are a conjunction of polyhedra.
Given an initial state $x_0\in\mathcal{X}$ and a sequence of control inputs $\boldsymbol{u}=(u_0, \ldots,u_{T-1})$, the sequence of states $\boldsymbol{x}=(x_0, \ldots,x_{T})$, which we call a \textit{trajectory}, is uniquely generated.

\subsection{Specification by Signal temporal logic}

In this study, we consider the specifications given in STL. 
STL is a predicate logic defined for specifying properties for continuous signals \cite{Fainekos2009-yn}. STL consists of predicates $\mu$ that can be obtained through a function $g^\mu(\cdot)$ as follows:
\begin{equation}
\mu=\begin{cases}
\top \text { if } g^\mu(\boldsymbol{x}) \leq 0 \\
\perp \text { if } g^\mu(\boldsymbol{x})>0.
\end{cases}
\end{equation}
In this study, we consider affine functions for $g^\mu$ of the form $g^\mu=a^{\mathsf{T}}x_t-b: \mathcal{X} \rightarrow\mathbb{R}$, where $a\in\mathbb{R}^{n}$ and $b\in\mathbb{R}$.
In addition to standard boolean operators $\wedge, \vee$, in STL, temporal operators $\square$ (\textit{always}), $\Diamond$ (\textit{eventually}), and $\boldsymbol{U}$ (\textit{until}) are also considered. Each temporal operator has an associated bounded time interval $[t_1,t_2]$ where $0\leq t_1<t_2<\infty$. We assume that STL formulae are written in negation normal form (NNF) without loss of generality \cite{sadradd_monotone}. STL formulae in NNF have negations only in front of the predicates \cite{Sadraddini2015-dn} and we omit such negations from STL syntax. We refer to this negation-free STL as STL, which is defined by \cite{Baier2008-up}:
$$
\begin{aligned}
\varphi:=\mu| \vee \varphi | \wedge \varphi|\square_{\left[t_1, t_2\right]} \varphi| \Diamond_{\left[t_1, t_2\right]} \varphi \mid \varphi_1 \boldsymbol{U}_{\left[t_1, t_2\right]} \varphi_2.
\end{aligned}
$$

The notion of \textit{robustness} is a useful semantic defined for STL formulae, which is a real-valued function that describes how much a trajectory satisfies an STL formula. Let $(\boldsymbol{x}, t)$ denote a trajectory starting at timestep $t$, i.e., $(\boldsymbol{x}, t)=x_t,x_{t+1},...,x_T$. The trajectory length $T$ should be selected large enough to calculate the satisfaction of a formula (see e.g., \cite{Sadraddini2015-dn}). 
Given an STL formula $\varphi$, we define the \textit{reversed} robustness with respect to a trajectory $\boldsymbol{x}$ and a time $t$ that can be obtained recursively according to the following quantitative semantics:

\begin{defi}\textit{(reversed-robustness)}\label{def:revrobustness}
\begin{subequations}\label{eq:robustnessrev}
\begin{align}
    & \rho_{\text{rev}}^\mu((\boldsymbol{x}, t))=g^\mu\left(x_t\right) \label{eq:reversed_predicate}\\
    & \rho_{\text{rev}}^{\varphi_1 \wedge \varphi_2}((\boldsymbol{x}, t))=\max \left(\rho_{\text{rev}}^{\varphi_1}((\boldsymbol{x}, t)), \rho_{\text{rev}}^{\varphi_2}((\boldsymbol{x}, t))\right) \label{eq:reversed_conj}\\
    & \rho_{\text{rev}}^{\varphi_1 \vee \varphi_2}((\boldsymbol{x}, t))=\min \left(\rho_{\text{rev}}^{\varphi_1}((\boldsymbol{x}, t)), \rho_{\text{rev}}^{\varphi_2}((\boldsymbol{x}, t))\right)\\
    & \rho_{\text{rev}}^{\square_{\left[t_1, t_2\right] \varphi}}((\boldsymbol{x}, t))=\max _{t^{\prime} \in\left[t+t_1, t+t_2\right]}\left(\rho_{\text{rev}}^{\varphi}\left(\left(\boldsymbol{x}, t^{\prime}\right)\right)\right)\\
    & \rho_{\text{rev}}^{\Diamond_{\left[t_1, t_2\right] \varphi}}((\boldsymbol{x}, t))=\min _{t^{\prime} \in\left[t+t_1, t+t_2\right]}\left(\rho_{\text{rev}}^{\varphi}\left(\left(\boldsymbol{x}, t^{\prime}\right)\right)\right)\\
    & \rho_{\text{rev}}^{\varphi_1 \boldsymbol{U}_{\left[t_1, t_2\right]} \varphi_2}((\boldsymbol{x}, t))=\max _{t^{\prime} \in\left[t+t_1, t+t_2\right]}\biggl(\min \Bigl(\Bigl[\rho_{\text{rev}}^{\varphi_1}(\left(\boldsymbol{x}, t^{\prime}\right)), \nonumber\\
    &  \min _{t^{\prime \prime} \in[t+t_1, t^{\prime}]}\bigl(\rho_{\text{rev}}^{\varphi_2}((\boldsymbol{x}, t^{\prime \prime}))\bigr)\Bigr]\Bigr)\biggr) \label{eq:reversed_until}
\end{align}
\end{subequations}
\end{defi}
Note that we define the recursive semantics for the \textit{reversed} robustness function, i.e., the minus of the original robustness $\rho_{\text{orig}}^{\varphi}$ defined in the literature (see e.g., \cite[Definition 2]{Belta_undated-jj}) as $\rho_{\text{rev}}^{\varphi}=-\rho_{\text{orig}}^{\varphi}$. We reverse the sign of the robustness functions $\rho_{\text{rev}}^{\mu}$ of predicates $\mu$, and we interchange the $\max$ and $\min$ operators. The reversed-robustness function is introduced as a notational simplification within our robustness decomposition framework. This approach removes unfavorable minus sign symbols in the decomposition and establishes a correspondence between convex parts in the optimization and conjunctive operators in the specification, and similarly for concave parts. It is worth mentioning that this modification does not change all the properties of the original robustness as we simply reverse the signs in the definition of the robustness function.

STL can be represented as a tree \cite{Kurtz2022-pe,Sun2022-fg}. An example is shown below. 
\begin{example}\label{ex:target}
This example deals with a specification called two-target specifications, which is borrowed from \cite{Kurtz2022-pe}. The formula of this specification is given as:
\begin{equation}
\varphi = \Diamond_{[0, T-T_d]}\left(\square_{[0,T_d]} B_1 \vee \square_{[0,T_d]} B_2\right) \wedge \square_{[0, T]} \neg O \wedge \Diamond_{[0, T]} G. \label{either_spec}
\end{equation}

Figure \ref{fig:either-MIP} shows the graphical representation of the state trajectory satisfying the two-target specification with $T=25, T_d=5$. The robot must remain in one of the two regions, either $B_1$ or $B_2$, for a dwelling time of $T_d$, and it must reach the goal region $G$ within the trajectory length of $T$ while avoiding the obstacle region $O$.

Figure \ref{fig:tree_two_target} shows a tree description of the two-target specification for reversed-robustness functions $\rho^{\varphi}$ with $T=1$ and $T_d=1$. 
While logically equivalent STL formulas can have different tree descriptions, we adopt the \textit{simplified form} of the tree description, which is uniquely determined per STL formula (for details, see \cite{Kurtz2022-pe}). 
The tree nodes marked ``LP" at the bottom represent linear predicates, while the top (outermost) node is marked ``and." 
The depth of this tree is four and the height of the predicates is either ``and-or-and-predicate" or ``and-or-predicate" regardless of the value of $T$ and $T_d$. Note that although this figure represents the specification with $T=1, T_d=1$ for visibility, the depth and height of the tree do not change even when we vary these parameters because changing them will only increase the number of predicates in each depth.
\end{example}

\subsection{Smooth approximation}\label{subsec:smooth}

The robustness function that results from \eqref{eq:robustnessrev} is not differentiable due to the $\operatorname{\max}$ and $\operatorname{\min}$ operators, and thus, we cannot calculate the gradients of this function. 
Recent papers such as \cite{Pant_smooth} smooth these functions in order to apply gradient-based methods. However, because smoothing both the $\max$ and $\min$ of the function leads to a problem that deviates from the original, it is preferable to minimize the parts that are smoothed. Our method requires to smooth \textit{only} $\min$ functions. This is because $\max$ functions in the robustness are rather decomposed in our method as shown in Section \ref{section:decomposition}. The popular smooth approximation of the $\max$ and $\min$ operators is by using the log-sum-exp (LSE) function. This approximation for $\min$ operators is given as:
\begin{equation}
\overline{\min }_k(a_1,...,a_r):=-\frac{1}{k} \ln \sum_{i=1}^r e^{-k a_i},\label{LSEmin}
\end{equation}
where $a_i\in\mathbb{R}$ and $k$ is the smooth parameter. Now, we define a new robustness measure that will be used throughout this paper. 
\begin{defi}\textit{(Smoothed reversed-robustness)}\label{def:newrobustness}
The smoothed reversed-robustness $\overline{\rho}^{\varphi}_{\text{rev}}(\boldsymbol{x})$ is defined by the quantitative semantics \eqref{eq:robustnessrev} where every $\min$ operator is replaced by $\overline{\min}_k$.
\end{defi}

\subsection{Problem statement}
This study considers the problem of synthesizing the trajectory that satisfies the STL specification in a maximally robust way. 
Given the system \eqref{eq:system} and the specification $\varphi$ and the initial state $x_0$, the optimization problem is as follows:
\begin{subequations}\label{eq:moto}
\begin{align}
\min _{\boldsymbol{x}, \boldsymbol{u}} \hspace{0.5em}&\quad \overline{\rho}_{\text{rev}}^{\varphi}(\boldsymbol{x}) \label{eq:min_x1}\\
\text { s.t. } & x_{t+1}=A x_t+B u_t \label{eq:min_sys1}\\
& x_t \in \mathcal{X}, u_t \in \mathcal{U} \label{eq:min_xu1}\\
& \overline{\rho}_{\text{rev}}^{\varphi}(\boldsymbol{x})\leq 0
\label{eq:min_xuu1}
\end{align}
\end{subequations}

Note that the reversed-robustness $\rho_{\text{rev}}^{\varphi}(\boldsymbol{x})$ is replaced with the smoothed reversed-robustness $\overline{\rho}_{\text{rev}}^{\varphi}(\boldsymbol{x})$ in \eqref{eq:min_x1} and \eqref{eq:min_xuu1}. Thus, although a feasible solution of the non-smoothed problem always satisfies the specification \eqref{either_spec}, a feasible solution of \eqref{eq:moto} does not necessarily satisfy the specification (as it under-approximates the non-smoothed robustness). However, the approximation error can be made arbitrarily small by making the smooth parameter $k$ sufficiently large \cite{Pant_smooth,Gilpin2021-wv}.

\section{A structure-aware decomposition of STL formulae}\label{section:decomposition}

In this section, we present our robustness decomposition framework. The goal of our robustness decomposition is to transform the problem in order to apply a convex-concave procedure (CCP)-based algorithm. 

\subsection{Convex-concave procedure (CCP)}\label{subsec:ccp}
The convex-concave procedure \cite{Lipp2016-fa,Shen2016-sn} is a heuristic method for finding a local optimum of a class of optimization problems called the difference of convex (DC) programs, which is defined as follows: 
\begin{defi}\textit{(Difference of convex (DC) programs)}
\begin{subequations}\label{eq:dc}
\begin{align}
&\operatorname{minimize} \quad f_0(\boldsymbol{z})-g_0(\boldsymbol{z}) \label{eq:dc1}\\
&\text{subject to } \quad f_i(\boldsymbol{z})-g_i(\boldsymbol{z}) \leq 0, \quad i\in\{1, \ldots, m\},\label{eq:dc2}
\end{align}
\end{subequations}
where $\boldsymbol{z}\in\mathbb{R}^h$ is the vector of $h$ optimization variables and the functions $f_i: \mathbb{R}^h \rightarrow \mathbb{R}$ and $g_i: \mathbb{R}^h \rightarrow \mathbb{R}$ are convex for $i\in\{0, \ldots, m\}$. The DC problem \eqref{eq:dc} can also include equality constraints 
\begin{equation}\label{eq:dc3}
p_i(\boldsymbol{z})=q_i(\boldsymbol{z}),
\end{equation}
where $p_i$ and $q_i$ are convex. These equality constraints are expressed as the pair of inequality constraints
\begin{align}\label{eq:dc4}
p_i(\boldsymbol{z})-q_i(\boldsymbol{z}) \leq 0, \quad q_i(\boldsymbol{z})-p_i(\boldsymbol{z}) \leq 0.
\end{align}
\end{defi}

The basic idea of CCP is simple; at each iteration of the optimization of the DC program \eqref{eq:dc}, we replace concave terms with a convex upper bound obtained by linearization and then solve the resulting convex problem using convex optimization algorithms.
It can be shown that all of the iterates are feasible if the starting point lies within the feasible set and the objective value converges (possibly to negative infinity, see \cite{Lipp2016-fa,Sriperumbudur_lanckriet_2009} for proof).

Note that either of the functions $p_i(\boldsymbol{z}), q_i(\boldsymbol{z})$ in equality constraints \eqref{eq:dc3} are also approximated by the CCP. This is because an equality constraint is transformed into two inequalities \eqref{eq:dc4} and both convex functions can make either inequality concave. Therefore, we should avoid introducing equality constraints if possible.

\subsection{Robustness decomposition}\label{subsection:first_step_of_}
Our decomposition framework alters the problem into an efficient form of a DC program considering the aforementioned features of CCP-based algorithms. This decomposition not only transforms the problem into a DC program but also enhances the efficiency of the overall algorithm. As a result, the resulting program approximates only the truly concave parts, and our approach results in a new SQP method. 

To this end, we need to know whether each function in the program \eqref{eq:moto} is convex or concave. 
In this program, the robustness function $\overline{\rho}_{\text{rev}}^{\varphi}$ is the only part whose curvature (convex or concave) cannot be determined.
Thus, if the robustness function $\overline{\rho}_{\text{rev}}^{\varphi}$, which is a composite function consisting of $\max$ and $\min$ operators, can be rewritten as a combination of convex and concave functions, the entire program can be reduced to a DC program.
The proposed framework achieves this by recursively decomposing the outermost operator of the robustness function until all the arguments of the functions are affine. As this study considers only affine functions for predicates $g^{\mu}(\cdot)$, if the arguments of a function are all predicates, the convexity or concavity of the function can be determined by examining the operator preceding it, which may be either a $\max$ or $\min$ function.
The recursive decomposition of a nonconvex robustness function into predicates corresponds to a decomposition of the robustness tree from the top to the bottom, in accordance with the tree structure.

Our method utilizes the idea of the epigraphic reformulation methods. This method transforms the convex (usually linear) cost function into the corresponding epigraph constraints. However, our approach is different from the normal epigraphic reformulation in the sense that we deal with \textit{nonconvex} cost functions and we have to apply the reformulation \textit{recursively}. Our robustness decomposition transforms \eqref{eq:moto} into a more efficient form of DC program. To demonstrate our main idea, for the remainder of this paper, we use the simplified two-target specification described in Example \ref{ex:target} as the specification $\varphi$.

As the robustness function $\overline{\rho}_{\text{rev}}^{\varphi}$ is in the cost function, the first procedure is to decompose the function from the cost function into a set of constraints. 
As we consider the two-target specification, the outermost operator of the robustness function $\overline{\rho}_{\text{rev}}^{\varphi}$ is $\max$, i.e., $\overline{\rho}_{\text{rev}}^{\varphi} = \max(\overline{\rho}_{\text{rev}}^{\varphi_1},\overline{\rho}_{\text{rev}}^{\varphi_2},...,\overline{\rho}_{\text{rev}}^{\varphi_r})$. We introduce a new variable $s_\xi$, and reformulate the program as follows.
\vspace{-0.5cm}
\begin{subequations}\label{eq:max_trans}
\begin{align}
\min _{\boldsymbol{x}, \boldsymbol{u},s_\xi}\hspace{0.2em} &s_\xi \label{eq:min_xi}\\
\text { s.t. } & x_{t+1}=A x_t+B u_t \label{eq:x}\\
& x_t \in \mathcal{X}, u_t \in \mathcal{U} \label{eq:xi_}\\
& s_\xi \leq 0   \label{eq:xi_0}\\
& \overline{\rho}_{\text{rev}}^{\varphi_1}\left(\boldsymbol{x}\right) \leq s_\xi \dots \overline{\rho}_{\text{rev}}^{\varphi_r}\left(\boldsymbol{x}\right)
\leq s_\xi \label{eq:xl}
\end{align}
\end{subequations}
Note that the set of inequalities \eqref{eq:xl} is equivalent to \begin{equation}\label{eq:eq1}
\max(\overline{\rho}_{\text{rev}}^{\varphi_1},\overline{\rho}_{\text{rev}}^{\varphi_2},...,\overline{\rho}_{\text{rev}}^{\varphi_r}) \leq s_\xi
\end{equation}
although the original semantics of the variable $s_\xi$ is the \textit{equality} constraint:
\begin{equation}\label{eq:eq}
\max(\overline{\rho}_{\text{rev}}^{\varphi_1},\overline{\rho}_{\text{rev}}^{\varphi_2},...,\overline{\rho}_{\text{rev}}^{\varphi_r}) = s_\xi.
\end{equation}
However, it is well known that this kind of $\max$ transformation is known to be equivalent, particularly for the convex case. Although \eqref{eq:max_trans} is nonconvex, we can also show that both 
\eqref{eq:moto} and \eqref{eq:max_trans} are equivalent (see \cite{Lindemann2019-os} and \cite[Section 8.3.4.4]{El_Ghaoui2015-eb}).
The proof of this equivalence is provided in Appendix \ref{app:maxtran} for a later explanation (in particular \eqref{eq:ii_minmax}).

As each $\overline{\rho}_{\text{rev}}^{\varphi_i}$ for $i=1,\ldots,r$ is a robustness function, each inequality in \eqref{eq:xl} can be restated as a constraint in one of the following two forms, depending on whether the outermost operator is $\max$ or $\min$:
\begin{align}
\max(\overline{\rho}_{\text{rev}}^{\Phi_1},...,\overline{\rho}_{\text{rev}}^{\Phi_{y_{\max}}})\leq s_{\xi}, \label{eq:max_before}\\
\overline{\min}(\overline{\rho}_{\text{rev}}^{\Psi_1},...,\overline{\rho}_{\text{rev}}^{\Psi_{y_{\min}}}) \leq s_{\xi}, \label{eq:mellowmin} 
\end{align}
where functions $\overline{\rho}_{\text{rev}}^{\Phi_j} (j\in\{1,...,y_{\max}\})$ (resp. $\overline{\rho}_{\text{rev}}^{\Psi_j}(j\in\{1,...,y_{\min}\})$) are robustness functions associated with $\Phi_j$ (resp. $\Psi_j$), which are the subformulas of $\varphi_i (i\in\{1,...,r\})$.
Note that these arguments of the $\max$ and $\min$ functions are still not necessarily predicates. We continue the following steps until all the arguments in each function become the predicates.

For inequality constraints of the form \eqref{eq:max_before}, we transform it as follows:
\begin{equation}\label{eq:max_after}
\begin{aligned}[b]
\overline{\rho}_{\text{rev}}^{\Phi_1}\leq s_{\xi},\overline{\rho}_{\text{rev}}^{\Phi_2}\leq s_{\xi} ,..., \overline{\rho}_{\text{rev}}^{\Phi_{y_{\max}}}\leq s_{\xi}.
\end{aligned}
\end{equation}
This is of course equivalent to \eqref{eq:max_before}. 

For constraints of the forms \eqref{eq:mellowmin}, we first check whether each argument of the $\overline{\min}$ function is a predicate or not. If not, we replace such an argument with a new variable. Because we consider the simplified two-target specification, the outermost operator of the function $\overline{\rho}_{\text{rev}}^{\Psi_j} (j\in\{1,...,y_{\min}\})$ is $\max$, i.e., $\overline{\rho}_{\text{rev}}^{\Psi_j}=\max(\overline{\rho}_{\text{rev}}^{\psi_1},...,\overline{\rho}_{\text{rev}}^{\psi_h})$. We first replace function $\overline{\rho}_{\text{rev}}^{\Psi_j}$ in \eqref{eq:mellowmin} by a new variable $s_{\text{new}}$ as 
\begin{equation}\label{eq:min_slack}
\overline{\min}(\overline{\rho}_{\text{rev}}^{\Psi_1},...,s_{\text{new}},...,\overline{\rho}_{\text{rev}}^{\Psi_{y_{\min}}})\leq s_\xi.
\end{equation}
Then, we add the following constraints:
\begin{equation}\label{eq:ii_minmax}
\max(\overline{\rho}_{\text{rev}}^{\psi_1},...,\overline{\rho}_{\text{rev}}^{\psi_h}) \leq s_{\text{new}}.
\end{equation}
Although the inequality constraint \eqref{eq:ii_minmax} should be equality, this modification is justified by the similar discussion for \eqref{eq:xl} even though the equivalence of this transformation is not immediately obvious (see the proof of Theorem \ref{theo:global}).

All the inequalities in \eqref{eq:max_after} and \eqref{eq:ii_minmax} are either in the $\max$ form \eqref{eq:max_before} or the $\min$ form \eqref{eq:mellowmin}. Thus, subsequent to these steps, the constraints of the forms \eqref{eq:max_before}, \eqref{eq:mellowmin} are decomposed into several constraints in the same forms again.
We repeat this procedure until we reach the predicates, which is the bottom of the tree in Figure \ref{fig:tree_two_target}. When we finish the above recursive manipulations, \eqref{eq:moto} is transformed into an DC program as follows:
\begin{subequations}\label{eq:final}
\begin{align}
        \min _{\boldsymbol{x}, \boldsymbol{u},s_\xi,\boldsymbol{s}_{\max},\boldsymbol{s}_{\min}}
        \hspace{0.2em}&s_\xi \label{eq:final_xx}\\
        \text { s.t. } & x_{t+1}=A x_t+B u_t  \label{eq:final_xuu}\\
        & x_t \in \mathcal{X}, u_t \in \mathcal{U}\label{eq:final_xu} \\
        & s_\xi \leq 0,  \label{eq:final_0} \\
&
\left.\hspace{-1cm}
\begin{tabular}{l}
$\overline{\rho}_{\text{rev}}^{\lambda^{(1)}_1}\leq s^{(1)}_{\max},...,\overline{\rho}_{\text{rev}}^{\lambda^{(1)}_{y_{\max}}}\leq s^{(1)}_{\max}$  \\
$\quad\quad\vdots \hspace{2.5cm}\vdots $  \\
$\overline{\rho}_{\text{rev}}^{\lambda^{(v)}_1}\leq s^{(v)}_{\max},...,\overline{\rho}_{\text{rev}}^{\lambda^{(v)}_{y_{\max}}}\leq s^{(v)}_{\max}$  \\
\end{tabular}
\right\} \begin{tabular}{l}\text{from }$\max$\\
\text{functions} \end{tabular}\label{eq:finalmax}\\
&
\left.\hspace{-1cm}
\begin{tabular}{l}
$\overline{\min}(\overline{\rho}_{\text{rev}}^{\mu^{(1)}_1},...,\overline{\rho}_{\text{rev}}^{\mu^{(1)}_{y_{\min}}}) \leq s^{(1)}_{\min}$ \\
$\quad\quad\vdots$ \\
$\overline{\min}(\overline{\rho}_{\text{rev}}^{\mu^{(w)}_1},...,\overline{\rho}_{\text{rev}}^{\mu^{(w)}_{y_{\min}}}) \leq s^{(w)}_{\min}, $\\
\end{tabular}
\right\}  \begin{tabular}{l}\text{from }$\min$\\
\text{functions} \end{tabular}\label{eq:finalmin}
\end{align}
\end{subequations}
where all the formulas $\lambda_1^{(\cdot)},...,\lambda_{y_{\max}}^{(\cdot)},\mu_1^{(\cdot)},...,\mu_{y_{\min}}^{(\cdot)}$ in \eqref{eq:finalmax} and \eqref{eq:finalmin} are predicates and $\boldsymbol{s}_{\max} = (s^{(1)}_{\max},...,s^{(v)}_{\max})$ and $\boldsymbol{s}_{\min} = (s^{(1)}_{\min},...,s^{(w)}_{\min})$. Note that, with some abuse of notation, some variables in $\boldsymbol{s}_{\max},\boldsymbol{s}_{\min}$ in constraints \eqref{eq:finalmax},\eqref{eq:finalmin} may represent $s_\xi$ and $s_{\text{new}}$, and some arguments of the $\overline{\min}$ function in \eqref{eq:finalmin} are not affine predicates but the new variables.

Important properties of this program \eqref{eq:final} are stated below. 

\begin{prop}\label{lem:robust_vs_slack}
Let $\boldsymbol{z}'=(\boldsymbol{x}', \boldsymbol{u}',s'_\xi,\boldsymbol{s}'_{\max},\boldsymbol{s}'_{\min})$ denote a feasible solution for program \eqref{eq:final}. Then, $\overline{\rho}_{\text{rev}}^{\varphi}(\boldsymbol{x}')\leq 0$ holds.
\end{prop}
\begin{proof}
See Appendix \ref{app:robust_s_xi}.
\end{proof}

\begin{prop}\label{lem:rev_feasible}
If $\boldsymbol{z}'$ is feasible for \eqref{eq:final}, $(\boldsymbol{x}',\boldsymbol{u}')$ is feasible for \eqref{eq:moto}.
\end{prop}
\begin{proof}
The inequality $\overline{\rho}_{\text{rev}}^{\varphi}(\boldsymbol{x}')\leq 0$ holds from Proposition \ref{lem:robust_vs_slack}, which means that \eqref{eq:min_xuu1} holds. Note that \eqref{eq:min_sys1} and \eqref{eq:min_xu1} are the same as \eqref{eq:final_xuu} and \eqref{eq:final_xu}. 
\end{proof}

\begin{theo}\label{theo:global}
The global optimum of two programs: \eqref{eq:moto} and \eqref{eq:final} is the same.
\end{theo}
\begin{proof}
See Appendix \ref{app:global}.
\end{proof}

\begin{figure*}[t]
 \begin{minipage}[b]{0.47\linewidth}
  \centering
  \includegraphics[keepaspectratio, scale=0.5]{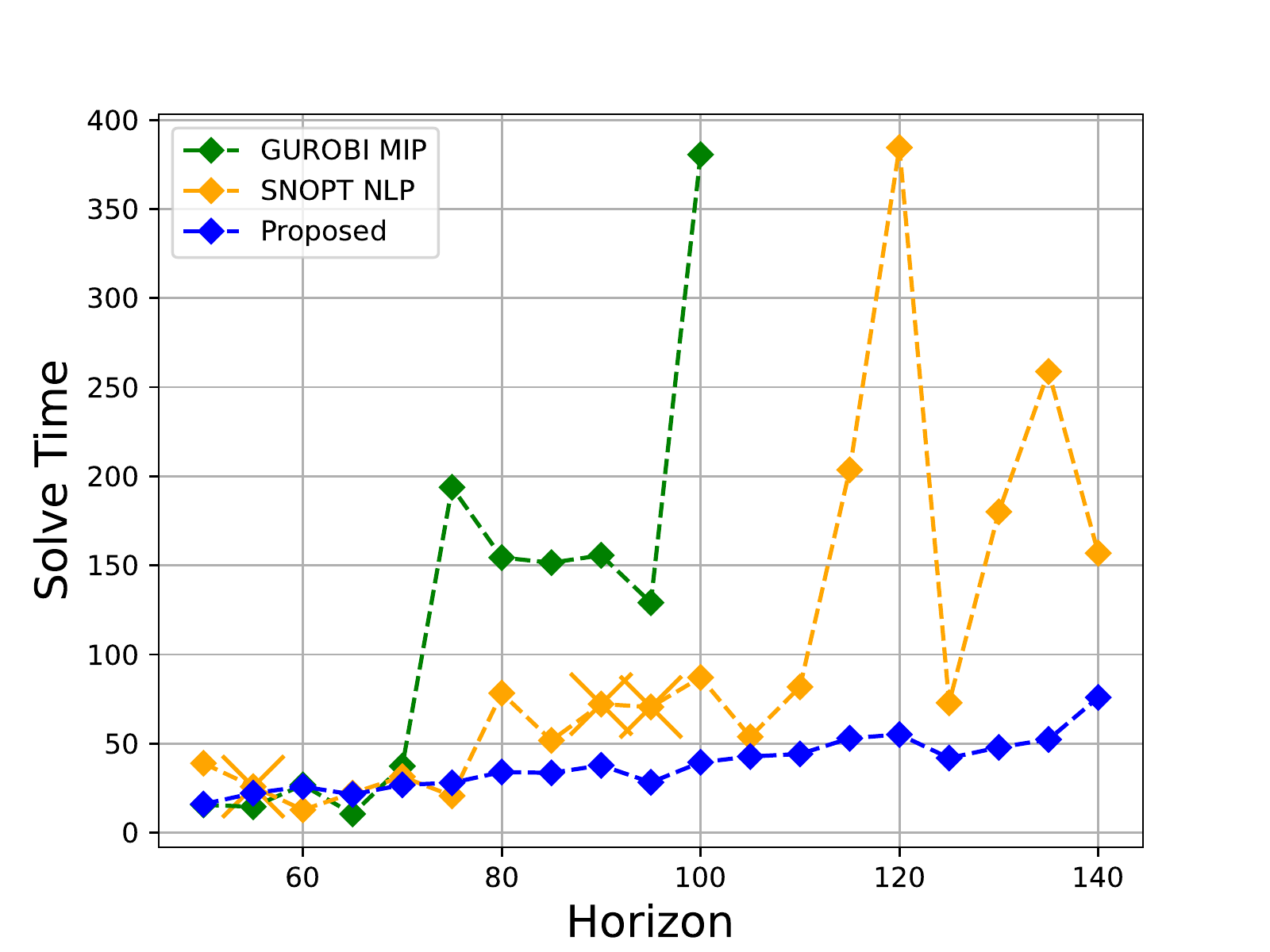}
  \subcaption{Computational times}\label{fig:solvetime_three}
 \end{minipage}
 \centering
  \begin{minipage}[b]{0.47\linewidth}
  \centering
 \mbox{\raisebox{0mm}{\includegraphics[keepaspectratio, scale=0.5]{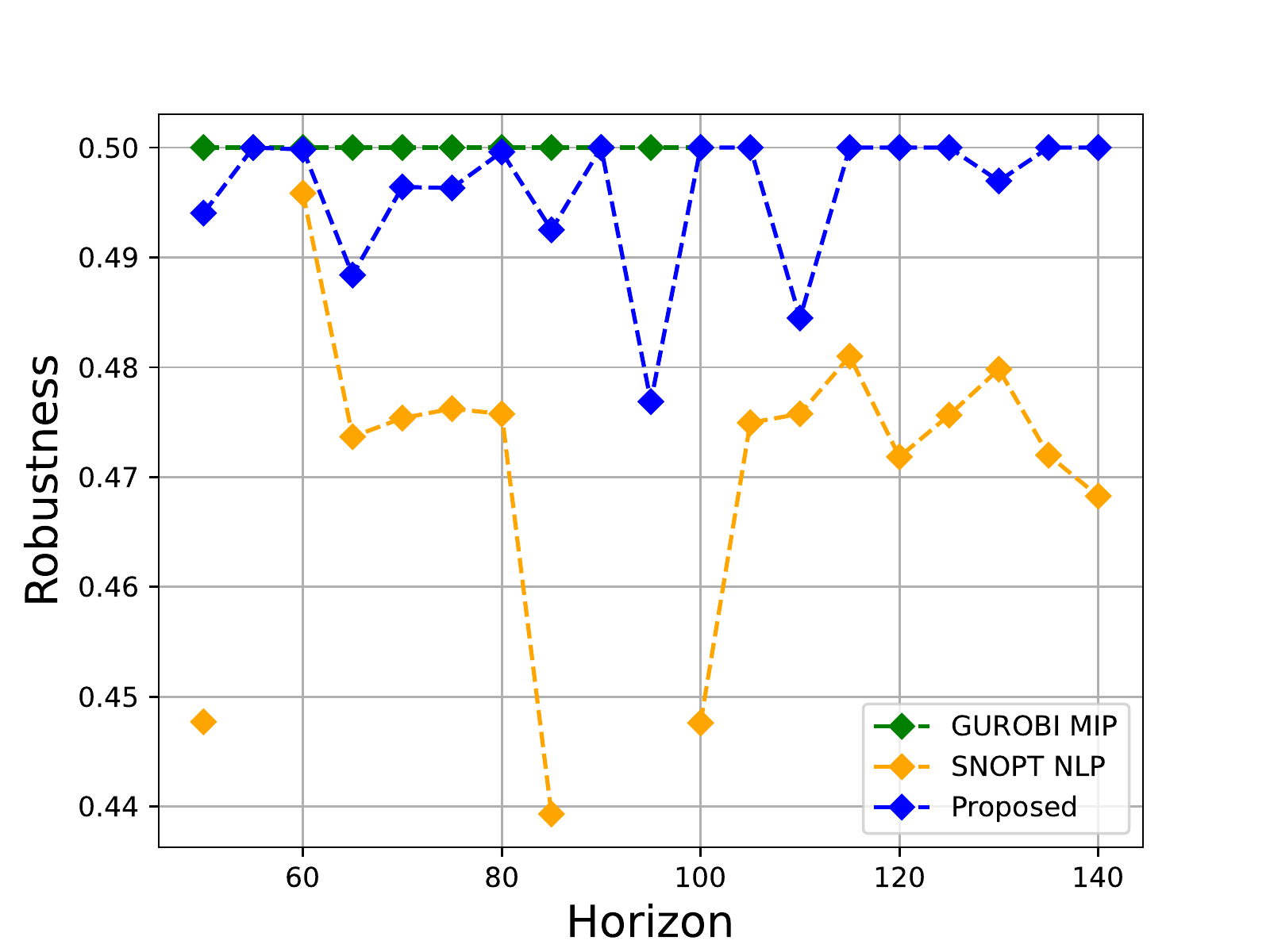}}}
  \subcaption{Robustness scores}
  \label{fig:robustness_three}
 \end{minipage}
 \caption{Computation time and robustness score of the three methods over every 5 horizons from $T=50$ to $140$.}\label{fig:three}
\end{figure*}

\section{Properties of the subproblem}\label{sec:subproblem}
After we obtain the transformed program \eqref{eq:final}, we apply the CCP method described in Subsection \ref{subsec:ccp}. We majorize the $\overline{\min}$ operator in \eqref{eq:finalmin} at the current point of each iteration. In general, the resulting program becomes a convex program such as a second-order cone program (SOCP) and CCP has to solve the convex program sequentially.
However, the subproblem we solve at each iteration is a linear program or a quadratic program, and not the general convex program as stated in the following theorem:
\begin{theo}\label{theo:convexquadratic}
The program \eqref{eq:final} after the majorization by CCP is a linear program (LP). If we add the quadratic cost function, the program can be written as a convex quadratic program (QP).
\end{theo}
\begin{proof}
In the program \eqref{eq:final}, the cost function is $s_\xi$ and the constraints including \eqref{eq:finalmax} are all affine except for concave constraints \eqref{eq:mellowmin}. As CCP linearizes all the concave constraints, the program resulting from the majorization of \eqref{eq:final} becomes an LP.
\end{proof}
Let $\mathcal{P}_{\text{LP}}$ denote this LP in the following.
It is worth mentioning that inequalities in \eqref{eq:finalmax} are all affine inequalities as we remove all the $\max$ operators in the end. As a corollary of the above theorem, our CCP-based method sequentially solves LP. If we add the quadratic cost function in \eqref{eq:final_xx}, our method becomes sequential quadratic programming (SQP). 

Moreover, as all the constraints coming from the robustness decomposition framework in Section \ref{section:decomposition} do not have any equality constraints, the non-convex parts of \eqref{eq:final} are only the inequality constraints \eqref{eq:finalmin}, which comes from $\min$ functions of the robustness function.

\begin{prop}\label{theo:convex}
The program \eqref{eq:final} has no equality constraints except for the state equations \eqref{eq:final_xuu}, which are affine constraints. 
\end{prop}

Hence, the CCP only approximates the \textit{truly} concave portions of \eqref{eq:finalmin} arising from the $\min$ functions in the robustness function. It is noteworthy that enforcing the constraints, such as \eqref{eq:xl} and \eqref{eq:ii_minmax}, as equality constraints would require the CCP to approximate them as described in Subsection \ref{subsec:ccp}, thereby leading to undesirable approximations. In contrast, the proposed method minimizes the number of parts requiring approximation, which enhances the algorithm's performance.

Furthermore, the feasible region of the program $\mathcal{P}_{\text{LP}}$ is interior to the feasible region of the program \eqref{eq:final}.

\begin{prop}\label{prop:feasible}
A feasible solution of the linear program $\mathcal{P}_{\text{LP}}$ is a feasible solution to the program \eqref{eq:final}.
\end{prop}
\begin{proof}
This is because the first-order approximations of CCP at each step are global over-estimators; i.e., $f_i(\boldsymbol{z})- g_i(\boldsymbol{z})\leq f_i(\boldsymbol{z})-g_i(\boldsymbol{z}_{(\cdot)})+ \nabla g_i(\boldsymbol{z}_{(\cdot)})^{\top}(\boldsymbol{z}_{(\cdot)}-\boldsymbol{z})$ where $\boldsymbol{z}_{(\cdot)}$ is the current point of variable $\boldsymbol{z}$.
\end{proof}

Finally, as a direct consequence of Propositions \ref{lem:robust_vs_slack} and \ref{prop:feasible}, the solution obtained by our CCP-based SQP method can always satisfy the specification, which is the most important property to guarantee safety.

\begin{theo}
Let $\boldsymbol{x}''$ denote a feasible trajectory of the program $\mathcal{P}_{\text{LP}}$. Then $\boldsymbol{x}''$ always satisfies the STL specification $\varphi$, i.e., $\boldsymbol{x}'' \vDash \varphi$ in the limit $k\rightarrow\infty$.
\end{theo}
\begin{proof}
Based on Propositions \ref{lem:robust_vs_slack} and \ref{prop:feasible}, $\overline{\rho}_{\text{rev}}^{\varphi}(\boldsymbol{x}'')\leq 0$ holds. The limit as $k\rightarrow\infty$ results in the smoothed reversed-robustness $\overline{\rho}_{\text{rev}}^{\varphi}(\boldsymbol{x}'')$ being equivalent to the original reversed-robustness $\rho_{\text{rev}}^{\varphi}(\boldsymbol{x}'')$. Therefore, $\rho_{\text{rev}}^{\varphi}(\boldsymbol{x}'')\leq 0$ as $k\rightarrow\infty$, which proves the statement.
\end{proof}
In practice, we select a sufficiently large $k$ to minimize the approximation error.

\section{Numerical Experiments}\label{section:example}
We illustrate the advantages of our proposed method through the two-target scenarios with varying time horizons ranging from $T=50$ to $140$.
All experiments were conducted on a 2020 MacBook Air with an Apple M1 processor and 8 GB of RAM\footnote{The source code is available at \url{https://github.com/yotakayama/STLCCP}.}. The state $x_t$ and input $u_t$ are defined as $
x_t=\left[p_x, p_y, \dot{p}_x, \dot{p}_y\right]^\mathsf{T},  u_t=\left[\ddot{p}_x, \ddot{p}_y\right]^\mathsf{T}
$, where $p_x$ is the horizontal position of the robot and $p_y$ is the vertical position. As for the system dynamics, we use a double integrator, i.e., a system \eqref{eq:system} with the matrix
\begin{align}
A&=\left[\begin{array}{ll}
I_2 & I_2 \\
0_{2\times 2} & I_2
\end{array}\right], \quad B=\left[\begin{array}{c}
0_{2\times 2} \\
I_2
\end{array}\right].
\end{align}

We implemented our algorithm in Python using CVXPY \cite{Diamond2016-ka} as the interface to the optimizer.
We chose the GUROBI(ver. 10)'s QP-solver with the default option. We took the penalty CCP \cite{Shen2016-sn} to solve the problem, which can admit the violation of constraints during the optimization procedures by penalizing the violation with variables. We introduced this penalty variable only for the concave constraints because the CCP only approximates such concave parts. Additionally, we added a small quadratic cost function. The weights on the penalty variables and the quadratic cost functions were 50.0 and 0.001, respectively. The resulting cost function \eqref{eq:final_xx} is represented as: 
\begin{equation}
s_\xi+50.0\sum_{i=0}^{w}(s_{i})+0.001\sum_{t=0}^T (x_t^\mathsf{T} Q x_t+u_t^\mathsf{T} R u_t),
\end{equation}
where $Q=\operatorname{diag}(0,0,1,1), R=I_2$ are positive semidefinite symmetric matrices, and $s_i$ for $i\in\{1,...,w\}$ ($w$ is the number of constraints in \eqref{eq:finalmin}) is the penalty variable added to the r.h.s. of each concave constraints \eqref{eq:finalmin} satisfying $s_i\leq 0$. The initial state is fixed as $x_0=[2.0,2.0,0,0]$, and the bounds on the state and input variables are $\mathcal{X}=[0.0, 0.0, -1.0, -1.0]^{\mathsf{T}}<x_t<[10.0, 10.0, 1.0, 1.0]^{\mathsf{T}}$ and $\mathcal{U}=[-0.2,-0.2]^{\mathsf{T}}<u_t<[0.2, 0.2]^{\mathsf{T}}$, respectively.

We compared our method with a popular MIP-based method and a state-of-the-art NLP-based method. For the MIP-based approach, we formulated the problem as a MIP using the encoding framework in \cite{Belta_undated-jj} and solved the problem with the GUROBI solver (GUROBI is often the fastest MIP solver). 
However, for the NLP-based method, we used the naive sequential quadratic programming (SQP) approach \cite{GilMS05} using the smoothing method proposed in \cite{Gilpin2021-wv}. All the parameters of the two methods were the default settings. 

Figure \ref{fig:three} compares convergence time and robustness score of the three methods from $T=50$ to $140$. The SNOPT-NLP method (orange) failed to output a feasible trajectory at 3 horizons ($T=55,90,95$) out of the 19 horizons, which are shown by the orange cross marker $\color{orange} \times$ in the left plot, and their robustness scores have been cut off in the right plot as they are $-\infty$.
In contrast, the proposed method (blue) failed only 2 times out of the 95 total trials (5 trials at each horizon). The GUROBI-MIP method (green) is truncated at the horizon $T=100$ due to the computation time $7371.3$ s, which exceeds the range of the plotted region. Therefore, subsequent results for longer horizons are not displayed for the sake of visibility. The plots of our method are the averages of different trials varying the initial values of variables, whereas the other two plots are from one trial.
All the scores in the right plot are calculated using the original robustness function.

In summary, our proposed method outperformed the state-of-the-art SNOPT-NLP method with respect to both the robustness score and the computation time. Specifically, as shown in the left plot, the computation time of our proposed method remained at a consistently low level even when the horizon increased. In contrast, the other two methods displayed volatility or were infeasibly slow, particularly above $T=100$. Furthermore, the right plot demonstrates that our proposed method did not compromise the robustness score to shorten the computation time. On the contrary, our method consistently achieved a higher robustness score than the SNOPT-NLP method. At some horizons, our method achieved a score of $0.5$, with the global optimum obtained by the MIP-based method. Notably, the plot of our method represents the average of all the trials, indicating that the feasible trajectories our method outputs are all the global optimum. In contrast, the SNOPT-NLP method sometimes generated failure trajectories with a robustness score of $-\infty$ in some horizons (at $T=55,90,95$), and all of the robustness scores (orange) were lower than those of the proposed method (blue).

Finally, Figure \ref{fig:traj} shows the trajectories generated by the proposed method with $T=50$. Our method produced a range of satisfactory trajectories, each dependent on the initial values of the system's variables.

\begin{figure}[H]
\vspace{-0.4cm}
  \centering
 \mbox{\raisebox{0mm}{\includegraphics[keepaspectratio, scale=0.45]{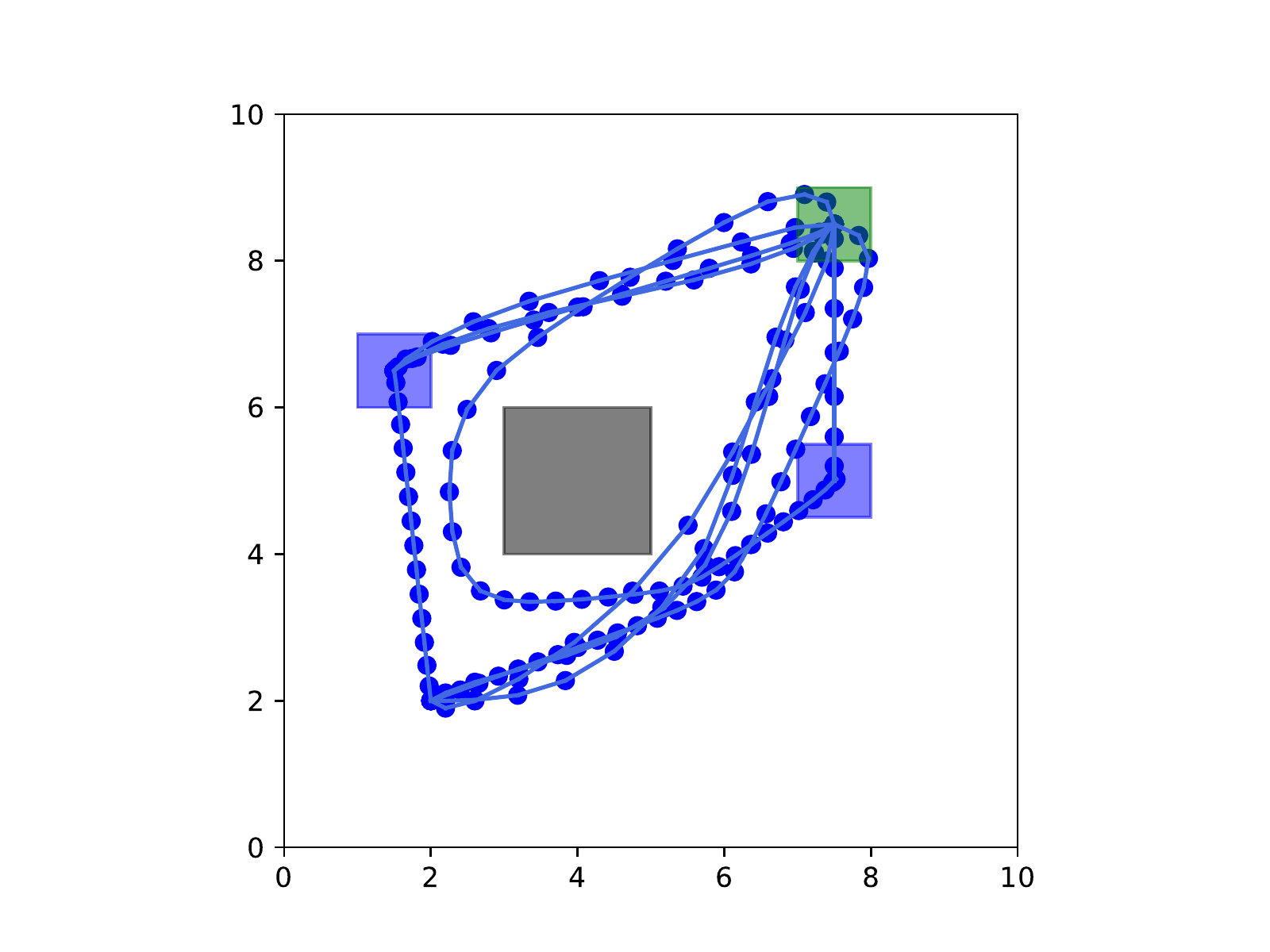}}}
     \caption{Example of trajectories generated by the proposed method with 5 random initial values for $T=50$. }\label{fig:traj}
 \end{figure}

\section{Conclusion}\label{section:conclusion}
In this study, we have introduced the CCP-based SQP algorithm for control problems involving STL specifications leveraging the inherent structures of STL. The subproblem of our proposed method is an efficient quadratic subprogram, achieved by approximating solely the concave constraints of the program that correspond to the disjunctive nodes of the STL tree. 

Future work includes extending the proposed method to non-affine predicates, nonlinear systems, and non-smooth cases. While this paper focused on linear systems to establish an efficient SQP approach, the same methodology can be applied to tackle more complex cases.

\bibliographystyle{IEEETran}

\appendix
\subsection{Equivalence between the two programs: \eqref{eq:moto} and \eqref{eq:max_trans} }\label{app:maxtran}
We prove the following equivalences.
\begin{prop}\label{prop:maxtrans}
Both formulations \eqref{eq:moto} and \eqref{eq:max_trans} are equivalent in the sense that the global optimum is the same.
\end{prop}

\begin{proof}
We demonstrate that if $(\boldsymbol{x}^*,\boldsymbol{u}^*)$ is optimal for \eqref{eq:moto}, then $(\boldsymbol{x}^*, \boldsymbol{u}^*,s_\xi^*=\overline{\rho}_{\text{rev}}^{\varphi}\left(\boldsymbol{x}^*\right))$ is optimal for \eqref{eq:max_trans}. Firstly, note that for any feasible solution $(\boldsymbol{x},\boldsymbol{u})$ of \eqref{eq:moto}, $(\boldsymbol{x},\boldsymbol{u},s_\xi=\overline{\rho}_{\text{rev}}^{\varphi}(\boldsymbol{x}))$ is feasible for \eqref{eq:max_trans} as \eqref{eq:xl} becomes an identity by this substitution $s_\xi=\overline{\rho}_{\text{rev}}^{\varphi}\left(\boldsymbol{x}\right)$. Thus, $(\boldsymbol{x}^*, \boldsymbol{u}^*,s_\xi^*=\overline{\rho}_{\text{rev}}^{\varphi}\left(\boldsymbol{x}^*\right))$ is also feasible for \eqref{eq:max_trans}. 

Suppose this solution $(\boldsymbol{x}^*, \boldsymbol{u}^*,s_\xi^*=\overline{\rho}_{\text{rev}}^{\varphi}\left(\boldsymbol{x}^*\right))$ is not optimal for \eqref{eq:max_trans}. Then, there exists another feasible solution $(\mathsf{x}^*,\mathsf{u}^*, \mathsf{s}^*_\xi)$ with a better objective, i.e., such that $\mathsf{s}^*_\xi<s_\xi^*=\overline{\rho}_{\text{rev}}^{\varphi}(\boldsymbol{x}^*)$. Since this feasible solution also satisfies \eqref{eq:xl}, that is, $\overline{\rho}_{\text{rev}}^{\varphi}\left(\mathsf{x}^*\right) \leq \mathsf{s}^*_\xi$, the inequality $\overline{\rho}_{\text{rev}}^{\varphi}\left(\mathsf{x}^*\right) < \overline{\rho}_{\text{rev}}^{\varphi}\left(\boldsymbol{x}^*\right)$ holds. Considering that $(\mathsf{x}^*,\mathsf{u}^*)$ is feasible for \eqref{eq:moto} (as one of the inequality conditions
in \eqref{eq:xl} must become the equality condition when the cost function $s_\xi$ takes a local minimum), the inequality $\overline{\rho}_{\text{rev}}^{\varphi}\left(\mathsf{x}^*\right) < \overline{\rho}_{\text{rev}}^{\varphi}\left(\boldsymbol{x}^*\right)$ contradicts the optimality of $\boldsymbol{x}^*$ for \eqref{eq:moto}. Therefore, $(\boldsymbol{x}^*, \boldsymbol{u}^*,\boldsymbol{s}_\xi^*=\overline{\rho}_{\text{rev}}^{\varphi}\left(\boldsymbol{x}^*\right))$ is optimal for \eqref{eq:max_trans}, and the global optimum is the same. 
\end{proof}

\subsection{Proof of Proposition \ref{lem:robust_vs_slack}}\label{app:robust_s_xi}
\begin{proof} 
We present a proof sketch. Let us first consider the initial replacement of the robustness function $\overline{\rho}_{\text{rev}}^{\Psi_i}(\boldsymbol{x}')$ with a new variable $s_{\text{new}}$. This transformation converts inequality \eqref{eq:mellowmin} into two inequalities, \eqref{eq:min_slack} and \eqref{eq:ii_minmax}. Note that \eqref{eq:ii_minmax} represents the inequality 
\begin{equation}\label{eq:childnode}
\overline{\rho}_{\text{rev}}^{\Psi_i}(\boldsymbol{x}')\leq s_{\text{new}},
\end{equation}
although this inequality sign should be an equality if $s_{\text{new}}$ were a substituting variable. Nevertheless, due to the fact that $\overline{\min}$ is an increasing function with respect to each argument, we have the inequality 
$$
\hspace{-0.15cm}\overline{\min}(\overline{\rho}_{\text{rev}}^{\Psi_1},...,\overline{\rho}_{\text{rev}}^{\Psi_i},...,\overline{\rho}_{\text{rev}}^{\Psi_{y_{\min}}}) \leq \overline{\min}(\overline{\rho}_{\text{rev}}^{\Psi_1},...,s_{\text{new}},...,\overline{\rho}_{\text{rev}}^{\Psi_{y_{\min}}}),
$$
Combining this inequality with \eqref{eq:min_slack}, we conclude that inequality \eqref{eq:mellowmin} holds even after the transformation with $s_{\text{new}}$. Similar analogies can be used for all subsequent transformations of non-convex robustness functions with new variables that follow the same pattern. Thus, ultimately, the original inequalities in \eqref{eq:xl} hold. By this inequality and \eqref{eq:final_0}, we have $\overline{\rho}_{\text{rev}}^{\varphi}(\boldsymbol{x}')\leq 0$.
\end{proof}

From a tree perspective, this proof demonstrates that by replacing the robustness of any subformula with a new variable that is greater than or equal to the original value (as shown in \eqref{eq:childnode}), the parent node's robustness score, due to the monotonicity of the $\max$ and $\min$ functions, also becomes greater than or equal to the original value (as shown in \eqref{eq:mellowmin}). This property extends further to the values of the grandparents and subsequent nodes, ultimately resulting in the cost function $s_{\text{new}}$, representing the top-node value, being greater than the robustness value of $\varphi$.


\subsection{Proof of Theorem \ref{theo:global}}\label{app:global}
We first prove the following proposition, which states the converse of Proposition \ref{lem:rev_feasible}:
\begin{prop}\label{lem:forward_feasible}
Let $\Phi^{(i)}_{\text{parent}}$ denotes the parent node of formulas $\Phi^{(i)}_1,...,\Phi^{(i)}_{y_{\max}}$ for $i\in\{1,...,v\}$, and $\Psi^{(i)}_{\text{parent}}$ denotes the parent node of formulas $\Psi^{(i)}_1,...,\Psi^{(i)}_{y_{\max}}$ for $i\in\{1,...,w\}$. Let $(\boldsymbol{x}^*,\boldsymbol{u}^*)$ denote a feasible solution of \eqref{eq:moto}, then the following solution $\boldsymbol{z}^*$ is feasible for \eqref{eq:final}.
\begin{align}\label{z_star}
&\boldsymbol{z}^*:=(\boldsymbol{x}^*,\boldsymbol{u}^*,s_\xi^*=\overline{\rho}_{\text{rev}}^{\varphi}(\boldsymbol{x}^*),
\boldsymbol{s}_{\max}^*=(\overline{\rho}_{\text{rev}}^{\Phi^{(1)}_{\text{parent}}}(\boldsymbol{x}^*),\nonumber \\ 
&...,\overline{\rho}_{\text{rev}}^{\Phi^{(v)}_{\text{parent}}}(\boldsymbol{x}^*)),\boldsymbol{s}_{\min}^*=(\overline{\rho}_{\text{rev}}^{\Psi^{(1)}_{\text{parent}}}(\boldsymbol{x}^*),...,\overline{\rho}_{\text{rev}}^{\Psi^{(w)}_{\text{parent}}}(\boldsymbol{x}^*)))
\end{align}
\end{prop}
\begin{proof}
When all the inequalities in \eqref{eq:finalmax} and \eqref{eq:finalmin} becoming identity equations, the feasible region of \eqref{eq:final} becomes identical to that of \eqref{eq:moto}, i.e., $S=\operatorname{Proj}(S_{\text{DC}})$ where $S$ and $S_{\text{DC}}$ denote the feasible region of \eqref{eq:moto} and \eqref{eq:final} respectively, and the operator $\operatorname{Proj}$ denotes the projection mapping by the substitution $\boldsymbol{z}^*(\cdot)$ from $(\boldsymbol{x},\boldsymbol{u},s_\xi,\boldsymbol{s}_{\max},\boldsymbol{s}_{\min})$ to $(\boldsymbol{x},\boldsymbol{u})$. Therefore, the feasible region of \eqref{eq:moto} is included in that of \eqref{eq:final}.
\end{proof}

We now show the following optimality equivalence: 
\begin{prop}\label{prop:forward_optimal}
If a feasible solution ($\boldsymbol{x}^*,\boldsymbol{u}^*$) is optimal for \eqref{eq:moto}, then $\boldsymbol{z}^*$ in \eqref{z_star} is optimal for \eqref{eq:final}. 
\end{prop}
\begin{proof} 
From Proposition \ref{lem:forward_feasible}, if $\boldsymbol{z}^*$ is not the optimal solution for \eqref{eq:final}, there exists another feasible solution $(\boldsymbol{x}',\boldsymbol{u}',s_\xi', \boldsymbol{s}_{\max}',\boldsymbol{s}_{\min}')$, with a better objective value, i.e., such that $s_\xi'< s_\xi^*=\overline{\rho}_{\text{rev}}^{\varphi}(\boldsymbol{x}^*)$. As in the proof of Proposition \ref{lem:robust_vs_slack}, $\overline{\rho}_{\text{rev}}^{\varphi}(\boldsymbol{x}')\leq s_\xi'$. Therefore, $\overline{\rho}_{\text{rev}}^{\varphi}(\boldsymbol{x}')<\overline{\rho}_{\text{rev}}^{\varphi}(\boldsymbol{x}^*)$. However, as Proposition \ref{lem:rev_feasible} shows that $\boldsymbol{x}'$ is also feasible for \eqref{eq:moto}, this inequality contradicts the fact that $\boldsymbol{x}^*$ is optimal for \eqref{eq:moto}. Therefore, $\boldsymbol{z}^*$ is optimal for \eqref{eq:final}.
\end{proof}

From this proposition, the global optimum of the two programs \eqref{eq:moto} and \eqref{eq:final} is the same, that is, $\overline{\rho}_{\text{rev}}^{\varphi}(\boldsymbol{x}^*)$. Moreover, we can also prove analogously that if a feasible solution $\boldsymbol{z}'$ is optimal for \eqref{eq:final}, then $(\boldsymbol{x}',\boldsymbol{u}')$ is optimal for \eqref{eq:moto}.
\end{document}